\documentclass[12pt,a4paper]{amsart}
\usepackage{latexsym,amssymb,amsmath,amsthm,amsfonts,enumerate,verbatim,xspace,
exscale}
\usepackage{graphicx}
\usepackage{color,amsbsy,textcomp}
\usepackage[gen]{eurosym}
\usepackage{float} 
 
\usepackage{hyperref}
\hypersetup{
    colorlinks=true,
    linkcolor=blue,
    filecolor=magenta,      
    urlcolor=cyan,
}

\theoremstyle{plain}
\newtheorem{theorem}{Theorem}[section]

\newtheorem{proposition}[theorem]{Proposition}

\newtheorem{criterion}[theorem]{Criterion}

\theoremstyle{definition}

\newtheorem{remark}[theorem]{Remark}

\theoremstyle{definition}
\newtheorem{ass}{Assumption}[section]

\newcommand{\up}{\upshape}

\def\vv<#1>{\langle#1\rangle}

\newcommand{\todo}[1]{}

\newcommand{\reviseF}[1]{#1} %
\newcommand{\reviseS}[1]{#1} %

\newcommand{\retwoS}[1]{#1} 

\usepackage[normalem]{ulem}
\usepackage{xcolor}

\newcommand\deleteF{\bgroup\markoverwith{\textcolor{blue}{\rule[0.8ex]{2pt}{0.9pt}}}\ULon}
\newcommand\deleteS{\bgroup\markoverwith{\textcolor{green}{\rule[0.8ex]{2pt}{0.9pt}}}\ULon}

\title[Analytical Validation Formulas]%
{Analytical Validation Formulas for Best Estimate Calculation in Traditional Life Insurance}

\author[Simon Hochgerner, Florian Gach]{Simon Hochgerner, Florian Gach}
\address{Austrian Financial Market Authority (FMA),
  Otto-Wagner Platz 5, A-1090 Vienna
} 
\email{simon.hochgerner@gmail.com}
\email{florian.gach@gmail.com}
\date{\retwoS{June 27, 2019}}

\begin{document}

\begin{abstract}
Within the context of traditional life insurance, a model-independent relationship about how the market value of assets is attributed to the best estimate, the value of in-force business and tax is established. This relationship holds true for any portfolio under run-off assumptions and can be used for the validation of models set up for Solvency~II best estimate calculation. Furthermore, we derive a lower bound for the value of future discretionary benefits. 
\reviseS{This lower bound formula is applied to publicly available insurance data to show how it can be used for practical validation purposes.}
\end{abstract}

\maketitle


\section{Introduction}
Per 1 January 2016 the European Union has implemented a new insurance regulation scheme called Solvency~II.  The own funds in this regime,  essentially, equal the excess of the market value of assets over the market  value of liabilities. \reviseS{The legal basis for calculating market values in accordance with Solvency~II is Art.~88 in \cite{L1}. One of the reasons that this is a difficult matter is that insurance liabilities are not traded in an active market. See~\cite{SS04,OBrien09,VKLP17,Delong,Laurent16} for more information and a discussion.}

In the context of traditional life insurance products, with a well-defined profit sharing mechanism, the market value of liabilities is dominated by  the \reviseS{so-called} best estimate.  
Indeed, for the aggregated Austrian market we have, per year end 2016, the following situation (in billion EUR):
\begin{center}
  \begin{tabular}{ l | c }
    \hline
    Technical provisions –- life (excluding health and index-linked and unit-linked) & 	 59   \\ 
    \hline
    Best estimate &	 58  \\ 
    \hline
    Risk margin &	 1  \\
    \hline
  \end{tabular}
\end{center}
That is, the risk margin is only about $1.8\%$ of the best estimate.
For Germany we have a similar situation:
\begin{center}
  \begin{tabular}{ l | c }
    \hline
    Technical provisions –- life (excluding health and index-linked and unit-linked) & 	 909    \\ 
    \hline
    Best estimate &	  894   \\ 
    \hline
    Risk margin &	 15  \\
    \hline
  \end{tabular}
\end{center}
Again, the risk margin amounts to $1.6\%$ of the best estimate. All these numbers are taken from the EIOPA insurance statistics web-site \cite{eiopa_is}.

Thus, own funds of the company strongly depend on the value of the best estimate ($BE$). Unfortunately, best estimate calculation for traditional life insurance books is a very difficult problem and its result generally subject to considerable modelling uncertainty.

To get an idea of the impact of errors, or uncertainty, of  $BE$ calculation, consider \retwoS{the following}, not untypical, sand-box life insurance example in Table~\ref{table:intro}.
\begin{table}[ht]
  \label{table:intro}
\centering
  \begin{tabular}{ l | c }
    \hline
    Best Estimate &	 3'864  \\ \hline
  Eligible Own Funds &	 572  \\  \hline
  Solvency Capital Requirement ($SCR$) &	 292  \\  \hline
  \end{tabular}
  \vspace{3pt}\caption{}
  \end{table}
In fact, these numbers are not arbitrary, but represent (in billion EUR) the European aggregate values taken from \cite{LTGA} for the base-line scenario. 
The corresponding solvency ratio is thus $572/292 \cong 196\%$. 
There are two (not unrelated) points to be made here: The eligible own funds are one \reviseF{order of} magnitude smaller than the technical provisions, which we have identified with the best estimate \reviseS{(by disregarding the much smaller risk margin)} for the sake of exposition. Secondly, an increase of the best estimate by only 
$1\%$ would result in a decrease of solvency ratio by $0.01\cdot 3'864/292 \cong 13$ percentage points.
An increase by $10\%$ would yield a solvency ratio below $100\%$, i.e.\ insolvency.
Here we assumed the $SCR$ to remain unchanged by an increase in $BE$. In reality the $SCR$ would increase with increasing $BE$, whence the resulting solvency ratio would be even smaller. 

In fact, \cite{LTGA} has considered mainly, but not exclusively, life insurance companies. If one were to include only traditional life insurance books in the sample, the impact of best estimate changes on solvency ratios should be expected to be even stronger.

We  \reviseS{mention} these numbers in order to stress the importance of reliable and stable best estimate calculation.

In order to calculate a best estimate for traditional life insurance in accordance with Solvency II regulation, one has to set up an asset liability model which generates cash-flows according to \reviseF{local generally accepted} accounting principles. In particular, such a model has to keep track of book and market values of assets under management and generate a return on assets that is in line with realistic management actions. Furthermore, a stochastic economic scenario generator is needed to model the asset portfolio and the relevant discount rates. 

It is therefore not surprising that there are no closed formulas \reviseS{to calculate best estimate liabilities associated to traditional life insurance contracts.
Let us remark in this context, that one of the outstanding features of  \retwoS{the} traditional life insurance business is that there exists a mechanism of profit sharing between the firm's surplus and the policy holder (see~\cite{Gerber}).
}

The goal of this paper is to present two results, Propositions~\ref{prop:gutlim} and \ref{prop:fdbgen}, which provide simple and effective tools for best estimate validation. 
These validation tools are the next best thing to having a closed formula for best estimate calculation in the sense that  they are \emph{analytical} and \emph{model-independent}. That is, they have to hold for any best estimate calculation in the context of life insurance with profit sharing -- as long as certain assumptions are met.
These assumptions are explicitly spelled out and subsequently discussed in Sections~\ref{sub:ass1} and \ref{sub:fdb-ass}.
\reviseS{Let us remark that these assumptions are tailor-made to the Austrian and German life insurance markets. We have not attempted a comprehensive study concerning the validity of the assumptions  for other European markets.}

\reviseS{As an example, we apply the estimation formula of Proposition~\ref{prop:fdbgen} to publicly available data of the Allianz Lebensversicherungs AG.}

Conclusions concerning our results are discussed in Section~\ref{sec:conclusions}. 

\retwoS{
\textbf{Acknowledgements.} We thank the anonymous reviewers for their detailed and helpful input.}

\textbf{Disclaimer.} The views expressed in this article reflect the authors' personal opinions and do not necessarily coincide with those of the Austrian Financial Market Authority (FMA).

\section{A basic equation of traditional life insurance valuation}
\subsection{Traditional life insurance company}\label{sec:1A}
We consider a company selling traditional life insurance business. Traditional means in this context that there is a procedure of profit participation which is defined by \reviseF{local generally accepted} accounting principles. 
\reviseS{
The set of insurance contracts in the company's liability book shall be denoted by 
\[
 \mathcal{X}.
\]
Hence each $x\in\mathcal{X}$ corresponds to an  individual
policy holder. The types of contracts we consider in this context are those of \cite[Chapters 3 and 4]{Gerber}: life insurance and life annuities, subject to profit participation.}
Let us assume: 
\begin{ass}\label{ass_A1}
The liability book consists of two items:
\begin{enumerate}[\up (1)]
\item
$TR$, the technical reserves;
\item
$SF$, the surplus fund.
\end{enumerate}
\end{ass}
The technical reserves are made up of the individual statutory reserves\reviseS{,
by which we mean the contract specific reserves $TR_x$ with respect to \reviseF{local generally accepted} accounting standards}. 
Thus to each insurance contract there corresponds a well-defined \reviseS{technical} reserve \reviseS{$TR_x$,} and $TR$ is the sum of all the individual reserves. 
\reviseS{We emphasize that the technical reserves $TR_x$ are assumed to be calculated as explained by \cite{Gerber}, that is, with respect to a constant technical interest rate and with respect to prudentially chosen safety loadings.}
The surplus fund, on the other hand, is an additional balance sheet item that belongs to the collection of insured viewed as a whole, whence it belongs to the liability book. However, it cannot be assigned to any individual contracts. Its purpose is to smoothen the declaration of profit participation over time.

The book values 
\[
 BV(TR) = \sum_{x\in\mathcal{X}}BV(TR_x) 
\]
and $BV(SF)$ 
of $TR$ and $SF$ \reviseF{of the technical reserves and of the surplus fund}, respectively, are well-defined balance sheet items in the realm of the local generally accepted accounting principles (see \cite[\S 144]{VAG} for the Austrian version). The objective of this section is to derive a general and model-independent relation for the associated best estimate, which is an item of the company's Solvency II balance sheet. This derivation depends on a number of assumptions which are discussed below in Subsection~\ref{sub:ass1}.

\subsection{The basic equation}\label{sub:beq}
We work on a discrete time grid $0,\ldots,T$. In practice time will be measured in years and $T$ will be the time until the run-off of (a sufficiently large percentage of) the portfolio. Usually $T$ is $60$ years or even more.
All our stochastic processes are assumed to be adapted to some underlying filtered probability space satisfying the usual assumptions.

Assume we have fixed an interest rate model \reviseF{with corresponding bank account numeraire}
\[
 B_t,
 \quad t=0,\ldots,T,
 \quad
 B_0 = 1.
\]
Let $\mathbb{Q}$ be the risk neutral martingale measure associated to our interest rate model or, more generally, our model of the economy.
For $0\le t \le s$ let $P(t,s)$ be the time $t$-value of a pay-out of one unit of currency at time $s$. Then, for any standard interest rate model, it is true that 
\begin{equation}\label{e:na}
 B_t^{-1}P(t,s) \quad\textup{ is a $\mathbb{Q}$-martingale }
\end{equation}
with respect to the filtration underlying the interest rate model. 
For instance, this holds for short rate models or Libor Market Models. 
\reviseS{We refer to \cite{Filipovic} for more information on interest rate models.}
In fact, Condition~\eqref{e:na} is \reviseS{the no-arbitrage condition as in, e.g.,}
\cite[Equ.~(3.2)]{TW}. 
In terms of simple one-year forward rates $F_t$, the bank account is given as the roll-over investment $B_t = \Pi_{s=1}^t(1+F_{\reviseF{s-1}})$.  
The best estimate $BE$ associated to $\mathcal{X}$ is then defined as
\begin{equation}\label{e:be}
 BE = E_{\mathbb{Q}}\left[ \sum_{t=1}^{T} B_t^{-1}\sum_{x\in\mathcal{X}}cf_{x, t} \right]
\end{equation}
where $cf_{x, t}$ is the cash-flow at time $t$ generated by contract $x$.
The cash flow has to take into account all relevant premia, benefits and costs -- see \cite[Art.~28]{L2}.

Let 
\begin{equation}
\label{e:BV}
 BV_t = BV_t(TR) + BV_t(SF)
\end{equation}
denote the book value of the liability portfolio at time $t$. The liability portfolio is covered by assets whose book value is assumed to equal $BV_t$. That is, we make the following simplification:
\begin{ass}\label{ass_A2} 
\reviseF{The company does not have equity in its statutory balance sheet, i.e.\ with respect to \reviseF{local generally accepted} accounting rules.}
\end{ass}
\reviseF{In other words, Assumption~\ref{ass_A2} requires the book value of assets to be equal to the book value of liablities, which is, of course, a very unrealistic assumption.}
Its meaning is discussed in Subsection~\ref{sub:ass1} \reviseF{below} where we will see that no generality is lost. We denote the total market value of the company's assets at time $t$ by
\[
 MV_t.
\]
The market value of assets $MV_t$ will change from one period to the next due to market movements. The same is true, albeit in a more complicated manner, for the book value \reviseF{of assets,} $BV_t$: it changes due to coupon payments, dividend yield, \reviseF{etc., as well as due to} realization of unrealized gains (whose value depends again on market movements). The precise manner in which this happens depends on the company's management rules and valuation choices such as the lower of cost or market principle.\footnote{\reviseF{The lower of cost or market principle, when applied in its strict form, requires a company to depreciate the book value of an asset whenever its market value falls below its current book value. When applied in the mild form, the book value is depreciated only if the market value is expected to remain below its current book value for an extended period of time.}} Regardless of the specific management rules and choices, we denote the company's book value return on $BV_t$ by
\[
roa_{t+1}
\]
and emphasize that this is the book value return \emph{before} corporate tax. 

For the next statement we note that we call market value induced changes in book value those changes which follow \reviseF{from the application of valuation principles, such as} the lower of cost or market principle, or the realization of unrealized gains. 

\begin{criterion}[Monetary conservation principle]\label{crit:2}
  All changes in the book value of assets $BV_t$ are either due to a cash-flow or due to a market value induced change in book value. 
\end{criterion}
This property is  fundamental. It can be viewed as
\reviseS{a} no-leakage and self-financing property: up to cash-flows, changes in the book value of the asset portfolio can only be due to interest rate or other (such as: stock) market effects. The above criterion is important in practice as it provides a simple yet challenging validation test for the inspection of real models\reviseF{: see} Proposition~\ref{prop:gutlim}.

Let us further elaborate on the no-leakage statement. To this end, we denote by $tax_t$ the corporate tax and by $sh_t$ the shareholder gains at $t$. 
Now, since $cf_t = \sum_{x\in\mathcal{X}}cf_{x,t}$ already includes all policy holder and cost cash-flows, the no-leakage criterion amounts to
\begin{align}\label{e:cf1}
 BV_t &=  BV_{t-1} - cf_{t} - sh_{t} - tax_{t} + roa_t
\end{align}


Let us define the \emph{unexpected return} by
\[
 ur_t :=  roa_t - F_{t-1} BV_{t-1}
\]
 where
\[
 F_{t-1} = \frac{B_t}{B_{t-1}} - 1
\]
is the simple forward rate from $t-1$ to $t$ implied by the interest rate model.

\reviseS{
It follows that
\begin{align*}
  B_T^{-1}BV_T
  &= BV_0 + \sum_{t=1}^T\Big(B_t^{-1}BV_t-B_{t-1}^{-1}BV_{t-1}\Big) \\ 
  &= BV_0 + \sum_{t=1}^T\Big(B_t^{-1}(BV_{t-1} - cf_{t} - sh_{t} - tax_{t} + roa_t)
        - B_{t-1}^{-1}BV_{t-1}\Big) \\
  &= BV_0 + \sum_{t=1}^T\Big(
        (B_t^{-1}( - cf_{t} - sh_{t} - tax_{t})
        + B_t^{-1}(roa_t - F_{t-1}BV_{t-1})
    \Big) \\
  &= BV_0 + \sum_{t=1}^T\Big(
        (B_t^{-1}( - cf_{t} - sh_{t} - tax_{t})
        + B_t^{-1}ur_t
        \Big).
\end{align*}
}
Taking the expected value with respect to the risk-neutral measure $\mathbb{Q}$, this becomes
\begin{equation}\label{e:gutlim}
  BV_0 + E_{\mathbb{Q}}\left[ \sum_{t = 1}^T B_t^{-1} ur_t \right]
  = BE + VIF + TAX + E_{\mathbb{Q}}\left[ B_T^{-1}BV_T \right],
\end{equation}
where 
\[
  VIF = E_{\mathbb{Q}} \left[ \sum_{t = 1}^T B_t^{-1} sh_t \right]
\]
is the so-called \emph{value of in-force business} and
\[
	TAX = E_{\mathbb{Q}} \left[ \sum_{t = 1}^T B_t^{-1} tax_t \right]
\]
is the value of \reviseF{corporate} tax payments.
The value of in-force business is the model dependent part of the \emph{market consistent embedded value} $MCEV$ which is generally expressed as 
\[
 MCEV = VIF + FC_0,
\]
where $FC_0$ is the market value of free capital at time $t = 0$. The $MCEV$ is a measure for the shareholder to determine how well the money is invested. Up to required capital and associated frictional costs, this definition coincides with that of \cite{CFO}.

Let us call $UG_t := MV_t - BV_t$ the unrealized gains.

\begin{proposition}[Basic equation of market consistent valuation]\label{prop:gutlim}
Let $T$ be the projection horizon. Then 
\begin{equation}\label{e:gutlim2}
  BV_0 + UG_0 =  BE + VIF + TAX + E_{\mathbb{Q}}\Big[B_T^{-1} MV_T \Big].
\end{equation}
Moreover, if $BV_t(SF)$ is bounded by $BV_t(TR)$ then
$E_{\mathbb{Q}}\Big[B_T^{-1} MV_T \Big]/MV_0 \approx 0$.   
\end{proposition}

In this statement $\approx$ means equality for all practical valuation purposes.\footnote{\reviseS{As a rule of thumb the relative error should not exceed 1\textperthousand, otherwise the impact of potential error on own funds would be too large -- compare with the sand box example in Table~\ref{table:intro}.}} We remark that there is no advantage in defining the relation $\approx$  in a more mathematical manner. Practically the remainder term 
$E_{\mathbb{Q}}\Big[B_T^{-1} MV_T \Big]/MV_0$ should be of the same order as the result of the leakage test, which would be the difference of the two sides of Equation~\eqref{e:gutlim2}. 

\reviseS{
Laimer \cite{Laimer} has verified in her diploma thesis that Equation \eqref{e:gutlim2} does indeed hold with $E_{\mathbb{Q}}[B_T^{-1} MV_T ]/A_0 = 0$ up to numerical errors. To do so, she employed a best estimate calculation tool proprietary to the FMA Austria and used several concrete traditional life insurance portfolios.} 

\begin{proof}
Let $\mathcal{A}_t$ denote the company's set of assets under management at time $t$.  
It follows that 
$ur_t = \sum_{a\in\mathcal{A}_{t-1}} ur_{a, t}$ where $ur_{a, t}$ is the contribution stemming from asset $a$ and \reviseF{where} the sum is over all assets in the portfolio at time $t-1$.
Now, for any asset $a\in\mathcal{A}_0$, we have 
\begin{align*}
 \sum_{t=1}^T B_t^{-1}ur_{a, t} 
 &=  \sum_{t=1}^T B_t^{-1}\Big(cf_{a, t} + a^*_t - (1+F_{t-1})  a^*_{t - 1}\Big)\\
 &=  \sum_{t=1}^T \Big( B_t^{-1}cf_{a, t} +  B_t^{-1}a^*_{t} -  B_{t-1}^{-1} a^*_{t - 1} \Big)\\
 &=  \sum_{t=1}^{T} B_t^{-1}cf_{a, t} 
 + B_{T}^{-1} a^*_{T} - a^*_{0} \\
 &= 
 \reviseS{\sum_{t=1}^{T} B_t^{-1}cf_{a, t} 
 + B_{T}^{-1} a_{T} - a^*_{0} 
 - B_{T}^{-1}(a_{T} 
 - a^*_{T}) }
\end{align*}
where $a^*_{t}$ (resp.\ $a_t$) is the time $t$ book (resp.\ market) value of $a$  
\reviseF{and } $cf_{a,t}$ is the asset's cash-flow (coupon, dividend payment, etc.) at $t$.
\reviseS{If $a$ has a maturity $T_a$ within the projection horizon such that $T_a \le T$, then book and market value coincide at $T_a$ such that $a_{T_a} = a^*_{T_a}$, and $cf_{a, t}=a_t=a_t^*=0$ for $T_a< t \le T$.
On the other hand, if $T_a>T$, then  
$a_{T}  - a^*_{T} = UG_{a,T}$ are the unrealized gains of $a$ at time $T$.} \reviseS{Because of \eqref{e:na} we have}
\[
a_0 =
E_{\mathbb{Q}}\left[ \sum_{t = \reviseF{1}}^{T} B_t^{-1} cf_{a, t} + B_{T}^{-1}a_{T} \right],
\]
\reviseS{which} implies that
\begin{equation} \label{e: ug2}
 E_{\mathbb{Q}}\Big[\sum_{t=1}^T B_t^{-1}ur_{a, t} \Big] 
 = a_0 - a_0^* \reviseF{- E_{\mathbb{Q}}\left[B_{T}^{-1}(a_{T} - a^*_{T})\right]}
 = UG_{a, 0} \reviseS{- E_{\mathbb{Q}}\left[B_{T}^{-1}UG_{a,T}\right]}.
\end{equation}
Assets that are bought in the course of reinvestment at $t>0$ satisfy $UG_{a, t} = 0$\reviseS{, because book value and market value coincide at time of purchase}. \reviseF{It follows} that 
\[
  E_{\mathbb{Q}}\Big[\sum_{t=1}^T B_t^{-1}ur_t\Big] 
 =  E_{\mathbb{Q}}\Big[\sum_{t=1}^T\sum_{a\in\mathcal{A}_{t-1}} B_t^{-1}ur_{a, t} \Big] 
 = UG_0 
 \reviseS{-E_{\mathbb{Q}}\left[B_T^{-1}UG_T\right]},
\]
where $\mathcal{A}_t$ denotes the asset portfolio at time $t$, and the result is independent of the particular reinvestment strategy.
The statement now follows from \eqref{e:gutlim}.
\end{proof}

\subsection{Discussion of assumptions}\label{sub:ass1}
Assumption~\ref{ass_A1} is very generic. All we need at this point are well-defined book and market values for the liability side of the balance sheet. If the cash-flows in Equation~\eqref{e:be} depend on additional provisions, these should be accordingly added to $BV_0$ and $UG_0$ in \eqref{e:gutlim2}. 

Assumption~\ref{ass_A2} is only at first sight a strong constraint. Actual companies will hold strictly positive equity. The relevant position should be added to the left hand side of \eqref{e:gutlim2}. For the right hand side, however, the statement $E_{\mathbb{Q}}\Big[B_T^{-1} MV_T \Big]/MV_0 \approx 0$ will now not be true anymore. The position $MV_T$ will still contain own funds. For actual validation purposes regarding \eqref{e:gutlim2} it is thus advisable to keep track of equity separately.

\section{An analytic lower bound for future discretionary benefits}\label{sec:fdb}
The goal of this section is to derive an analytic, i.e.\ model-independent, formula for the value of future discretionary benefits.  \reviseF{This is achieved under certain assumptions which are} discussed in Section~\ref{sub:fdb-ass} below.

\subsection{The lower bound formula}
The best estimate can be written as
\[
BE = GB + FDB
\]
where $GB$ \reviseF{stands for \emph{guaranteed benefits} and} is the value of all future cash flows that are already guaranteed at time of calculation $t=0$. Thus $GB$ is independent from all future developments and therefore purely deterministic. In particular, its value is independent of all management actions and economic scenarios.
On the other hand, $FDB$ stands for \emph{future discretionary benefits} and denotes the value of those cash flows that arise via the (future) profit sharing mechanism.  Solvency~II requires that $FDB$ be  \reviseF{determined and reported as a stand-alone} part of the best estimate\reviseF{: see \cite[Art.~25]{L2}.}  

\begin{remark}
\label{rem:GB}
\reviseS{We emphasize that this does not mean that the guaranteed benefits are fixed from the policy holder's perspective. In reality, benefits could be influenced by time of surrender, time of death, interest rate movements, tax incentives or other unknown variables. However, the point is that these variables contribute to the guaranteed benefits (as defined by the reporting template \cite[Template~S.12.01.01]{L3templ}) with their \emph{expected values.} It is only in this sense, that the guaranteed benefits are model-independent and fixed cash-flows.}
\end{remark}

The profit sharing mechanism dictates that the collective of policy holders receives a yearly accounting flow $ph_t^*$. 

\reviseF{We make a few generic assumptions:\footnote{We reiterate that we only consider with-profit contracts. All assumptions are discussed in Section~\ref{sub:fdb-ass}.}}

\begin{ass}\label{ass_B1}
We assume that  \reviseF{the profit sharing} mechanism is clearly defined by legislature and stable management rules. 
\end{ass}
\begin{ass}\label{ass_B2}
The gross policy holder profit participation rate $gph$ is constant. This is the rate with which the policy holder participates in the company's declared gross surplus \reviseF{$gs^*$ (if positive).} It does \emph{not} say anything about the company's return.
\end{ass}
We emphasize that $ph_t^*$ is an accounting flow and \emph{not} a cash flow. \reviseF{As such it is not paid out to the policy holder at time $t$, but rather increases the book value of liabilities.} 
Observe that
\begin{align}\label{e:ph1}
    PH^*
    &:= E_{\mathbb{Q}}\Big[\sum B_t^{-1} ph_t^*\Big]\\\notag
    &= gph \;E_{\mathbb{Q}}\Big[\sum B_t^{-1} gs_+^*(t)\Big]\\\notag
    &\ge gph\; E_{\mathbb{Q}}\Big[\sum B_t^{-1} gs^*(t)\Big]\\\notag
    &= gph\; (VIF + PH^* + TAX )\notag
\end{align}
 \reviseF{where $sh_t$, $tax_t$ are the respective shareholder, tax cash flows, and $x_+=\max(x,0)$. Here we have used the splitting $gs_t^* = sh_t + ph_t^* + tax_t$ together with the definitions of $VIF$ and $TAX$ from Section~\ref{sub:beq}.} Note that $sh_t$ can be negative, which corresponds to the case of shareholder capital injection.

\begin{ass}\label{ass_B3}
The policy holder participation $ph_t^*$ is negatively correlated with discount rate movements\reviseF{; i.e.,} policy holder participation will generally increase when interest rates increase: $Corr[B_t^{-1},ph_t^*] < 0$.
\end{ass}
\begin{ass}\label{ass_B4}
\reviseS{We assume, for the purpose of this section, that the liability book consists of only one contract and that the time to maturity of this contract is $M$.} 
\end{ass}
\begin{ass}\label{ass_B5}
\reviseF{The technical reserves $TR_t$ evolve deterministically.} Insurance technical gains are deterministic\reviseF{; i.e.,} we do not consider stochastic mortality modelling or stochastic (and/or dynamical) surrender behavior.
\end{ass}

Notice that the future discretionary benefits \retwoS{received by the policy holder} at time of maturity $M$ are exactly the sum $\sum_{t\le M}ph_t^*$ of accumulated policy holder profits.
This is actually a tricky point and holds only because we assume that policy holder survival probabilities (mortality, surrender, etc.) have already been taken into account. At the same time we do not list this point as an assumption, because it only means that we regard cash flows of surviving policy holders.

Notice that
\begin{align}
  FDB
  &= 
  E_{\mathbb{Q}}
  \Big[ \reviseS{B_M^{-1}}\sum_{t\le M} ph_t^*\Big]\notag \\
  &= 
  E_{\mathbb{Q}}\Big[ B_M^{-1}\Big]
    \cdot E_{\mathbb{Q}}\Big[\sum_{t\le M}ph_t^*\Big]
    + Cov\Big[B_M^{-1}, \sum_{t\le M}ph_t^* \Big]
    \notag \\
  &\ge 
    \reviseS{P(0,M)}
    \cdot E_{\mathbb{Q}}\Big[\sum_{t\le M}ph_t^*\Big]
    - 1\cdot  SD\Big[B_M^{-1}\Big]\cdot SD\Big[\sum_{t\le M}ph_t^*\Big]
\label{e:ph2}
\end{align}
\reviseS{where 
$P(0,M) = E_{\mathbb{Q}}[ B_M^{-1}]$ is the discount factor at time $0$ and 
we have made use of the fact that the correlation  $Corr[B_M^{-1}, \sum_{t\le M}ph_t^*]$ is bounded from below by $-1$.}

\reviseS{
\begin{remark}[Standard deviation]\label{rem:SD}
For a random variable $X$, we shall denote the standard deviation by $SD[X] = \sqrt{E[(X-E[X])^2]}$.
In formula~\eqref{e:ph2} the standard deviation is understood with respect to the risk neutral measure $\mathbb{Q}$. 
Often the standard deviation is denoted by $\sigma$ when viewed as a parameter to be inferred from a financial or econometric model. 
We have chosen the notation $SD[\cdot]$ to reflect the purely statistical approach of Assumption~\ref{ass_B8}.  
Of course, the statistical approach could also be replaced by a parametric model. However, since formula~\eqref{e:eta} depends on the product of two standard deviations, the gain in accuracy of the lower bound~\eqref{e:fdb} by using  a more refined model to estimate $SD[\cdot]$ is  limited.  
\end{remark}
}

On the other hand,  $Corr[B_t^{-1},ph_t^*] < 0$ yields
\begin{align}\label{e:ph3}
  PH^*
  &= E_{\mathbb{Q}}\Big[\sum_{t\le T}B_t^{-1}ph_t^* \Big]
  \le \sum_{t\le T}E_{\mathbb{Q}}[B_t^{-1}] \cdot E_{\mathbb{Q}}[ph_t^*]
  \le 
  \max_{1\le t\le T}\reviseS{P(0,t)} 
  \cdot E_{\mathbb{Q}}\Big[\sum_{t\le T}ph_t^*\Big].
\end{align}
Note that, if forward rates are positive, this maximum is simply $\max_{1\le t\le M}\reviseS{P(0,t)} = (1+F_0)^{-1}$. 
Currently interest rates are negative at the short end. 
Nevertheless, \reviseS{for the EIOPA curve per year-end 2017 as displayed in Table~\ref{table:discount rates},} the term 
\reviseS{$\max_{1\le t\le M}P(0,t)\cong 1.005$} 
is very close to $1$, and we will simply set \reviseF{it equal to $1$} for better readability and because the error is negligible.  The point is, in any case, that this term is deterministic and can be calculated from the initial forward curve.

\begin{proposition}\label{prop:fdb}
Let $A_0=MV_0=BV_0+UG_0$, $A_T = E_{\mathbb{Q}}[B_T^{-1}MV_T]$ and
\begin{equation}
\eta :=
\reviseS{P(0,M)}
\cdot\Big(
1 - \frac{SD[B_M^{-1}]}{\retwoS{P(0,M)}}
\cdot
\frac{SD[\sum_{t<M}ph_t^*]}{E_{\reviseS{\mathbb{Q}}}[\sum_{t<M}ph_t^*]}
\Big)\cdot\frac{gph}{1-gph}
\label{e:eta}
\end{equation}
Then we have the lower bound
\begin{equation}\label{e:fdb}
  FDB \ge \frac{\eta}{1+\eta}(A_0 - A_T - GB)
\end{equation}
\end{proposition}

\begin{proof}
The inequality follows from equations~\eqref{e:ph1}, \eqref{e:ph2}, \eqref{e:ph3} and Proposition~\ref{prop:gutlim}.
\end{proof}

\reviseS{
\subsection{A lower bound formula for various maturities}\label{sec:fdb-gen}
For the purpose of this section, we shall keep all of the above assumptions except Assumption~\ref{ass_B4}. 
We replace the latter as follows:
\begin{ass}\label{ass_B4X}
Assets are not attributed to individual contracts. 
\end{ass}

Let $\mathcal{X}$ denote the (finite) set of all contracts in the liability book. For each $x\in\mathcal{X}$ define $\eta_x$ according to formula~\eqref{e:eta} with respect to the contract's maturity $M_x$. Further, let
\[
 D_x := \frac{\eta_x}{1+\eta_x}.
\]
Recall from Assumption~\ref{ass_A1} that the statutory technical reserves are the sum of the individual reserves, $BV(TR) = \sum_{x\in\mathcal{X}}BV(TR^x)$. Consider 
\begin{itemize}
    \item 
    $A_0^x := A_0 \cdot BV(TR^x)/BV(TR)$, which is a proportional attribution of market values to individual contracts according to the principle~\eqref{ass_B4X};\footnote{The attribution can only depend on the book value, since $TR_x$ is a statutory reserve and therefore does not have a market value.}
    \item
    $GB^x$ denotes the value of guaranteed benefits (at time $0$) of contract $x$, and $GB = \sum_{x\in\mathcal{X}}GB^x$;
    \item 
    $FDB = \sum_{x\in\mathcal{X}}FDB_x$.
\end{itemize}
\begin{ass}
\label{ass:FDBx}
Suppose that Proposition~\ref{prop:fdb} can be applied to each contract such that
\[
 FDB_x \ge D_x(A_0^x - A_{M_x}^x - GB^x).
\]
\end{ass}
Notice that we do not assume $A_{M_x}^x=0$. This is  to allow for cross-financing between contracts, after time $M_x$. Define the weights
\[ 
 w_x :=  \frac{A_0^x-GB^x}{A_0-GB}
\]
and the \emph{weighted depreciation factor}
\begin{equation}
    \label{e:D}
    D 
    := \sum_{x\in\mathcal{X}} w_x D_x.
\end{equation}
It follows that 
\begin{align*}
    FDB 
    &= \sum_{x\in\mathcal{X}}FDB_x 
    \ge \sum_{x\in\mathcal{X}}D_x(A_0^x-A_{M_x}^x-GB^x)\\
    &= D\sum_{x\in\mathcal{X}} (A_0^x-A_{M_x}^x-GB^x)
      - \sum_{x\in\mathcal{X}}(D-D_x)(A_0^x-A_{M_x}^x-GB^x) \\
    &= D(A_0 - GB) - \sum_{x\in\mathcal{X}}D_x A_{M_x}^x
\end{align*}
where we use that the weights $w_x$ are chosen such that $\sum_{x\in\mathcal{X}}(D-D_x)(A_0^x-GB^x) = 0$.
The quantities $A_{M_x}^x$ correspond to the fraction of $A_0^x$ which remains in the model after time $M_x$. Unfortunately, these quantities are model dependent and are, therefore, a priori unknown. The term  $\sum_{x\in\mathcal{X}}D_x A_{M_x}^x$ accounts for the cross-financing. If all of $A_0^x$ were to be accounted for by a cash-flow up to, and including, time $M_x$, then $A^x_{M_x}=0$. 

\begin{proposition}\label{prop:fdbgen}
Assume that there exists $F>0$ such that $\sum_{x\in\mathcal{X}}D_x A_{M_x}^x \le F$. Then
\begin{equation}
    \label{e:FDBgen}
    FDB \ge D(A_0 - GB) - F.
\end{equation}
\end{proposition}
Existence of such an $F$ means that cross-financing is bounded by an a priori estimated quantity $F$.
See Section~\ref{sub:fdb-num} for a concrete application of this formula.
}

\subsection{Concrete numbers}\label{sub:fdb-num}
\reviseF{
Let us apply formula \eqref{e:FDBgen} to publicly available data  from the Allianz Lebensversicherungs AG.
Allianz Lebensversicherungs AG is a German life insurance company which has profit sharing contracts in its liability book.
The data in Table~\ref{table:allianz} is taken from the publicly available reports \cite{GB} and \cite{SFCR}, which concern the accounting year 2017. Hence the applied interest rate information from  Table~\ref{table:spot rates} is also with respect to year-end 2017.  

\begin{table}[ht]
\centering
\begin{tabular}{cclll}
\multicolumn{1}{l}{\textbf{Symbol}} & \multicolumn{1}{c}{\textbf{Value}} & \textbf{Source} & \textbf{Name in source} &  \\ \cline{1-4}
\multicolumn{1}{|c|}{$BV_0$} & \multicolumn{1}{r|}{192.3} & \multicolumn{1}{c|}{\cite[p.~46]{SFCR}} & \multicolumn{1}{l|}{Versicherung mit \"Uberschussbeteiligung} &  \\ \cline{1-4}
\multicolumn{1}{|c|}{$UG_0$} & \multicolumn{1}{r|}{43.2} & \multicolumn{1}{c|}{\cite[p.~46]{GB}} & \multicolumn{1}{l|}{Stille Reserven der 
einzubeziehenden Kapitalanlagen} &  \\ \cline{1-4}
\multicolumn{1}{|c|}{$SF_0$} & \multicolumn{1}{r|}{10.4} & \multicolumn{1}{c|}{\cite[p.~52]{SFCR}} & \multicolumn{1}{l|}{\"Uberschussfonds} &  \\ \cline{1-4}
\multicolumn{1}{|c|}{$GB$} & \multicolumn{1}{r|}{154.1} & \multicolumn{1}{c|}{\cite[p.~46]{SFCR}} & \multicolumn{1}{l|}{Bester Sch\"atzwert: Wert f\"ur garantierte Leistungen} &  \\ \cline{1-4}
\multicolumn{1}{|c|}{$FDB$} & \multicolumn{1}{r|}{48.6} & \multicolumn{1}{c|}{\cite[p.~46]{SFCR}} & \multicolumn{1}{l|}{Bester Sch\"atzwert:  zuk\"unftige \"Uberschussbeteiligung} &  \\ 
\cline{1-4}
\end{tabular}%
\vspace{3pt}\caption{Allianz Lebensversicherungs AG: public data. Values are in billion Euros.}
\label{table:allianz}
\end{table}
}
\reviseS{
\begin{ass}\label{ass:C1}
Suppose the weighted depreciation factor~\eqref{e:D} is equal to the one with maturity $M = 15$, so that $D=\eta_{15}/(1+\eta_{15})$ 
with 
\begin{equation}
    \label{e:eta15}
    \eta_{15}
    :=
    P(0,15)
\cdot\Big(
1 - \frac{SD[B_{15}^{-1}]}{P(0,15)}
\cdot
\frac{SD[\sum_{t<15}ph_t^*]}{E_{\mathbb{Q}}[\sum_{t<15}ph_t^*]}
\Big)\cdot\frac{gph}{1-gph}.
\end{equation}
\end{ass}
Variations of this assumption are shown in Tables~\ref{table:lb10} and \ref{table:lb20}.

For concrete validation purposes, companies would have the data to explicitly calculate $D$ according to formula~\eqref{e:D}, since $w_x$ and $D_x$ are quantities which are known \emph{a priori}. 

\begin{ass}\label{ass:C1X}
Assumption~\ref{ass_B2} is made concrete by setting $gph = 0.80$. In Table~\ref{table:lb} we shall show results for $gph = 0.75$, $gph = 0.80$ and $gph = 0.85$.
\end{ass}
}

\begin{ass}
\label{ass_B6}
Surplus funds are bounded by the technical reserve, i.e., $SF_t \le \theta TR_t$ where $\theta >0$ is constant.
\end{ass}
\begin{ass}\label{ass_B7}
  The variance of $ph_t^*$ is not very high, that is, we  assume $SD[\sum ph_t^*] \cong 5\% \cdot E[\sum ph_t^*]$. 
\end{ass}
\begin{ass}\label{ass_B8}
\reviseS{I}nterest rate variance should also be reasonably bounded.  \reviseF{When estimated on monthly historical data 
from year-end 2014 to year 2017
as shown in Table~\ref{table:spot rates}, we find, for the  coefficient of variation, that  $SD[B_{15}^{-1}]/P(0,15) \cong 4\%$.}
\end{ass}


\reviseS{  
With the discount factor $P(0,15)\cong 84\%$ from Table~\ref{table:discount rates}, we insert these numbers in \eqref{e:eta} to obtain
\begin{align*}
    \eta 
    &\cong 84\%\cdot\Big(1 - 4\% \cdot 5\%\Big)\cdot \frac{0.8}{1-0.8}
    \cong 3.35
\end{align*}
and furthermore
\begin{equation}
    LB_1 :=
    \frac{\eta}{1+\eta}\Big(A_0 - GB\Big)
    \cong
    77\%\cdot\Big(192.3+43.2-154.1\Big)
    \cong
    62.68
\end{equation}
in billion Euros. To compare this number to  $FDB = 48.6$, we have to subtract two quantities from the lower bound $LB_1$:
\begin{enumerate}[\up (1)]
    \item 
    According to \cite[Art.~91]{L1} the surplus fund, $SF_0 = 10.4$, is not part of the Solvency~II value of liabilities \emph{if this article is authorized by national law}.  
    This is the case for Germany (see \cite{bafin15}), whence the surplus fund is to be subtracted from the future discretionary benefits which are calculated by the company (and this deductible is not part of the reported $FDB$). 
    \item
    The cross-financing term $F$ from Proposition~\ref{prop:fdbgen}.
\end{enumerate}
Therefore, the resulting lower bound is
\begin{equation}
    LB = LB_1 - SF_0 - F \cong 52.28 - F. 
\end{equation}
To estimate $F$, we have to attribute $A_0$ to individual contracts and say something about the cross-financing effects. This information is not publicly available. We 
therefore separate $A_0$ into buckets $A_0^{x(t)}$, where $A_0^{x(t)}$ belongs to those contracts $x(t)$ which mature at time $t$. Thus contracts are bundled according to their time of maturity. We have to make an assumption concerning the run-off of the portfolio:
\begin{ass}
\label{ass:C2}
The portfolio run-off is roughly geometric, so that the value of reserves  is reduced by a factor of $\frac{1}{2}$ every $10$ years until the end of the projection. 
This is formalized as the requirement
\begin{align*}
    A_0^{x(t)} 
    &:= 
    \Big(2^{-\frac{t-1}{10}} - 2^{-\frac{t}{10}}\Big)A_0
    \textup{ for } t=1,\dots,T-1 \\
    A_0^{x(T)}
    &:= 2^{-\frac{T-1}{10}} A_0,
\end{align*}
and we shall assume that $T=60$. 
\end{ass}
Notice that $A_0 = \sum_{x\in\mathcal{X}}A_0^x = \sum_{t=1}^{60}A_0^{x(t)}$ and 
the corresponding run-off is given by 
\begin{equation}
    \label{e:run-off}
    A_0 - \sum_{t=1}^{s}A_0^{x(t)} = A_0\cdot 2^{-s/10} 
\end{equation}
for $s<T$. 

Furthermore, we have to make an assumption concerning the policy holder cross-financing term $A^t_{M_t}$ from Proposition~\ref{prop:fdbgen} (where $M_t=t$). This term has to be a fraction of $A_0^{x(t)}$, but this fraction need not be the same for all contracts. For example, it is conceivable that contracts with a low technical interest rate will yield more cross-financing in comparison to contracts with a high technical interest rate. Moreover, the cross-financing effect need not be constant in time. 
\begin{ass}
\label{ass:C3}
The cross-financing factor is a decreasing function of time, but does not depend on other contract properties (such as technical interest rate): $A_t^{x(t)} = C(t)A_0^{x(t)}$ where  
\[
 C(t) = C_0\frac{T-t}{T},  
\] 
$T = 60$; and we will consider 
$C_0 = 1\%$, $C_0 = 3\%$ and $C_0 = 5\%$.
\end{ass}

With the above assumptions it is possible to  calculate the term $F$ from \eqref{e:FDBgen}. 
Indeed, we set 
\begin{equation}
    F 
    = \sum_{x\in\mathcal{X}}D_x A_t^x
    = C_0\sum_{t=1}^{60}\frac{\eta_{x(t)}}{1+\eta_{x(t)}}\frac{T-t}{T}A_0^{x(t)} 
\end{equation}
where 
\begin{equation}
    \eta_{x(t)}
    = P(0,t)\Big(1 - \frac{SD[B_t^{-1}]}{P(0,t)}\cdot 5\%\Big)\frac{gph}{1-gph}.
\end{equation}
Now, $P(0,t)$ is taken from Table~\ref{table:discount rates} and $SD[B_t^{-1}]$ is estimated from Table~\ref{table:spot rates}.
With $gph = 80\%$ and $C_0=3\%$ this yields $F\cong 4.1$, whence we obtain the lower bound
\begin{equation}\label{e:LB}
    LB = LB_1 - SF_0 - F \cong \euro{48.2}\,bn.
\end{equation}
This number should be compared with the value of $FDB = \euro{48.6}\,bn$ in Table~\ref{table:allianz}.
The assumptions leading to \eqref{e:LB} and, in particular, the choices $M=15$, $gph=80\%$ and $C_0=3\%$ are discussed in Section~\ref{sub:fdb-ass}.
See Section~\ref{sec:appendix} for further choices  of $M$, $gph$ and $C_0$. 
}

\subsection{Discussion of assumptions}\label{sub:fdb-ass}
\reviseS{In the following we shall provide justification for the above assumptions. Nevertheless, we stress that these remain unproved (in the precise mathematical sense) assumptions based on heuristic arguments and expert judgment. For some of the assumptions we can provide a sensitivity analysis in Section~\ref{sec:appendix}.}

Assumption~\ref{ass_B1} is one that is necessary for any asset liability model that could be employed for best estimate calculation. It is also necessary for our derivation of the lower bound formula. 
\reviseS{The \emph{net policy holder participation fraction} shall be denoted by 
\[
 nph.
\]}
Under Austrian law (\cite[\S~3]{GBVVU}) companies are required to share \emph{at least} $\reviseS{nph} = 85\%$ of their net profits  with policy holders. 
As the surplus fund  belongs to the liability side, this sharing mechanism does \emph{not} imply that $85\%$ of net profits are directly declared to specific policy holder accounts. Rather the profits are shared with the surplus fund and then may be used in the future, according to discretionary management rules, to increase policy holder profits. We also remark that the situation is very similar in Germany.
The Solvency II requirements for realistic modelling of future management actions are given in \cite[Art.~23]{L2}.

Assumption~\ref{ass_B2} means that the profit sharing rate and the \reviseF{corporate} tax rate are constant. 
In Austria this rate is $25\%$.  
\reviseS{If this rate is not constant, one would have to use a mean rate to derive the corresponding $gph$ from $nph$.}

Assumption~\ref{ass_B3} can be seen as a consequence of Solvency II's going concern hypothesis \cite[Art.~101]{L1}.  Indeed, if interest rates go up, one would assume that companies increase their policy holder profit declarations in order to remain a competitive participant in the market. 

\reviseS{Assumption~\ref{ass_B4} is only used to derive the preliminary result stated in Proposition~\ref{prop:fdb}, and is then removed in Section~\ref{sec:fdb-gen}.}

\reviseS{Assumption~\ref{ass_B5} remains unjustified as we do not know of any publicly available data to support it.} 

\reviseS{Assumption~\ref{ass_B4X} means that assets are shared equally among policy holders. For instance, this is the case for Germany (\cite[\S 3(1)]{MindZV}) 
and Austria (\cite[\S 3]{GBVVU}).}

\reviseS{Assumption~\ref{ass:FDBx}: Consider a contract $x$ with maturity $M_x$. This statement means that we assume the cross-financing between $x$ and other contracts, which occurs \emph{before} $M_x$, can be neglected \emph{on average}. Alternatively, one could separately consider the cross-financing before time $M_x$ and then assume that this cross-financing fully contributes to the $FDB$ in Equation~\eqref{e:FDBgen}, whence the lower bound remains unchanged.}

\reviseS{Assumption~\ref{ass:C1}: 
To calculate the weighted depreciation factor~\eqref{e:D}, one would need portfolio specific data. According to \cite[Fig.~A.~II.1]{LTGA} the distribution of liability duration for Germany is between 14 and 22 years. Since formula~\eqref{e:eta} includes a discounting term $P(0,t)$, we have chosen $M=15$ to reflect the fact that contributions from later times give rise to a lesser weight. The sensitivity with respect to  this assumption is shown in Tables~\ref{table:lb10} and \ref{table:lb20} for $M=10$ and $M=20$, respectively.}

\reviseS{Assumption~\ref{ass:C1X}: 
German legislature \cite[\S~4]{MindZV} distinguishes between two different values for the net policy holder participation $nph$. 
For \emph{most} sources of profit \cite[\S~4]{MindZV} dictates that $nph=90\%$. To arrive at the corresponding $gph$ one has to account for tax payments such that ``$nph$ times net profit'' equals ``$gph$ times gross profit''. 
This assumption thus amounts to using $gph = 80\%$ as an approximating average factor.\footnote{The corresponding Austrian value for $gph$ can be calculated to be  $80.75\%$.}  
The sensitivity on this assumption is shown in Tables~\ref{table:lb}, \ref{table:lb10} and \ref{table:lb20}. 
}

Assumption~\ref{ass_B6} 
means that \reviseF{gains as well as hidden reserves cannot be kept from the policy holder indefinitely.} 
This assumption is necessary in order to apply Proposition~\eqref{prop:gutlim} with  $E_{\mathbb{Q}}[B_T^{-1}MV_T] = 0$. 
\reviseS{For modelling purposes, one could apply the following  future management action: 
$SF_t \le \theta MV_t(TR)$ where the concrete value of $\theta<1$ is a management rule.}

Assumption~\ref{ass_B7} reflects a typical goal of management. In fact, the very purpose of the surplus fund is to keep policy holder participation stable over time. 
Further, we remark that \eqref{e:ph2} involves a product of two standard deviations. Thus even a higher value for $SD[\sum ph_t^*]$ would not have significant impact on \reviseS{the lower bound~\eqref{e:LB}.} 
Hence this assumption reflects the assumption that management would try to follow interest rate movements in profit sharing declaration but, at the same time, would try to avoid strong jumps in the declaration.

\reviseS{Assumption~\ref{ass_B8} is a consequence of the historical data in Tables~\ref{table:spot rates} and \ref{table:discount rates}. Nevertheless, we have listed this as an assumption since there are many different estimators and time series which one could use to find the coefficient of variation $SD[B_t^{-1}]/P(0,t)$. 
As with Assumption~\ref{ass_B7}, we remark that the sensitivity on this assumption is relatively low, since we are dealing with a product of two standard deviations.}

\reviseS{
Assumption~\ref{ass:C2} is a formalization of the idea that the portfolio run-off is, roughly, homogeneous in time: according to \cite[p.~20]{GB} the number of policies in the with-profit business went from approximately 10.5 million to approximately 10 million in the year 2017 (disregarding new business). This is a reduction of about  $5\%$. Now, time-homogeneity in this context should mean that the number of policies will be reduced by a factor of $(1-5\%)^t$ after $t$ years. Thus the run-off is geometric and we take this as a justification for the assumed portfolio reduction~\eqref{e:run-off}. 
However, in \eqref{e:run-off} we are concerned with the run-off of technical provisions instead of policy numbers. Therefore, this argument is only a partial justification for \eqref{e:run-off}.
We stress that companies would have the necessary information in this context, whence Assumption~\ref{ass:C2} would not be needed in real applications (or could be tested against real data). }

\reviseS{
Assumption~\ref{ass:C3} is related to Assumption~\ref{ass_B6}: if we did not assume that the surplus fund was bounded by the decreasing value of technical reserves, then the remainder term $A_t^{x(t)}$ could increase the surplus fund at time $t$ and never be converted to a cash-flow.
On the other hand, if we allow for an upper bound rule such as Assumption~\ref{ass_B6}, then $A_t^{x(t)}$ has to decrease with time. A good first approximation for $C_0$ could be $SF_0/A_0 = 4.4\%$. 
We have chosen $C_0=3\%$ to indicate that not all of the cross-financing $A_t^{x(t)}$ has to be kept within the surplus fund, but a part may also be directly declared to other policy holders (compare with Assumption~\ref{ass:FDBx}) at times $s\le t$. Sensitivity on this assumption is shown in Tables~\ref{table:lb}, \ref{table:lb10} and \ref{table:lb20}.
}

\section{Discussion of results}\label{sec:conclusions}

\subsection{The basic equation}
\reviseS{Best estimate calculation for with-profit life insurance products is a numerical task. This is due to the fact that there are no closed formula solutions to obtain a best estimate, whence Monte Carlo techniques are employed. This involves, in particular, setting up a full asset liability model such that all relevant cash-flows are generated. The general calculation process is discussed in \cite{Laurent16}. Given the degree of sophistication of these asset liability models, one should not expect a closed formula for best estimate calculation. 
}


The virtue of Equation~\eqref{e:gutlim2} is thus that it is a closed formula and has to hold for all asset liability models (that are market consistent and satisfy Assumptions \ref{ass_A1} and \ref{ass_A2}).

Of course, this does not mean that we can compute the best estimate solely based on Equation~\eqref{e:gutlim2}. The terms $VIF$ and $TAX$ are just as hard to compute as $BE$. The practical value of this equation rather is tied to the fact that it gives a very straightforward validation test for \emph{any} market consistent best estimate calculation model. 

This validation procedure can be viewed as a \emph{leakage test}. 
\reviseS{It is a necessary condition for two model properties:}
\begin{itemize}
\item 
All cash flows are accounted for -- nothing is ``lost'' by the numerical model.
\item
The (numerically generated) economic scenarios are free of arbitrage. 
\end{itemize}
It seems to us that this test is well-known to at least parts of the applied insurance mathematics community as a kind of ``folklore wisdom''. 
We remark here that we have found this formula independently. More importantly, we do not know of any literature which explicitly states this formula and, much less, of any references where the assumptions have been spelled out in a rigorous manner.

\subsection{The lower bound formula}
The best estimate can be \retwoS{expressed} as a sum of guaranteed cash flows $GB$ and future discretionary benefits $FDB$. Moreover, it is a Solvency II requirement to report $FDB$ separately as part of the best estimate (\cite[Template~S.12.01.01]{L3templ}).

The guaranteed benefits depend on second order assumptions concerning, e.g., mortality and surrender but are otherwise \emph{completely model-independent}. In particular, they do not depend on future management actions or economic scenarios. Indeed, these are fixed cash flows whence they can be valuated deterministically using only the initial risk-free (EIOPA-) interest rate curve.\footnote{See Remark~\ref{rem:GB}.}

On the other hand, the calculation of future discretionary benefits $FDB$ involves all the intricacies of best estimate valuation\reviseS{: asset liability model, management rules, economic scenarios, etc.\ -- see \cite{Kemp09,Laurent16,VKLP17}. Hence one should not expect a closed forumla solution for the $FDB$.}
However, the next best thing to a closed formula 
 is a lower bound, and this is what Proposition~\ref{prop:fdbgen} achieves. 

Clearly this lower bound cannot be used for reporting purposes, since the true $FDB$ could be considerably larger. We see three important applications for this estimate:
\begin{itemize}
\item 
It can be used by the company as an immediate test for their $FDB$ calculated by means of a numerical asset liabilty model.
\item
It can be used by the company as a target towards which to optimize the $FDB$ by \reviseS{an appropriate} choice of admissible management rules (reinvestment strategy, \reviseS{profit sharing, surplus fund management,} etc.)
\item
It can be used by the regulator as a simple on- or off-site plausibility check for reported $FDB$s.
\end{itemize}

All of this works only if the assumptions \reviseS{and generic management rules} discussed in Section~\ref{sub:fdb-ass} are either directly met or suitably amended. 

\reviseS{To show how Proposition~\ref{prop:fdbgen} could be used in practice we have, in Section~\ref{sub:fdb-num}, applied the lower bound formula to publicly available data of Allianz Lebensversicherung AG \cite{SFCR,GB}. In Equation~\eqref{e:LB} this yields a lower bound of $\euro{48.2}\,bn$, while the reported number, included in Table~\ref{table:allianz}, is $FDB = \euro{48.6}\,bn$. The assumptions leading to the result~\eqref{e:LB} are discussed in Section~\ref{sub:fdb-ass}. Finally, we have included results for the lower bound with respect to variations of some of the key assumptions in Tables~\ref{table:lb}, \ref{table:lb10} and \ref{table:lb20}.}

\reviseS{
\section{Appendix}\label{sec:appendix}
Table~\ref{table:spot rates} is used to calculate the standard deviation $SD[B_t^{-1}]$ of the observed discount factor $P(0,t)$ from Table~\ref{table:discount rates}. Compare with Remark~\ref{rem:SD}.

\begin{table}[H]
\centering
\begin{tabular}{cclll}
\cline{1-4}
\multicolumn{1}{|c|}{$M=15$} 
& \multicolumn{1}{c|}{$gph = 75\%$} & \multicolumn{1}{c|}{$gph = 80\%$} 
& \multicolumn{1}{l|}{$gph = 85\%$} &  \\ 
\cline{1-4}
\multicolumn{1}{|c|}{$C_0 = 1\%$} 
& \multicolumn{1}{c|}{46.9} & \multicolumn{1}{c|}{50.9} & \multicolumn{1}{c|}{55.7} &  \\ 
\cline{1-4}
\multicolumn{1}{|c|}{$C_0 = 3\%$} 
& \multicolumn{1}{c|}{44.3} 
& \multicolumn{1}{c|}{48.2} 
& \multicolumn{1}{c|}{52.8} &  \\ 
\cline{1-4}
\multicolumn{1}{|c|}{$C_0 = 5\%$} 
& \multicolumn{1}{c|}{41.8} & \multicolumn{1}{c|}{45.4} & 
\multicolumn{1}{c|}{49.8} &  \\
\cline{1-4}
\end{tabular}%
\vspace{3pt}\caption{Lower bound for the $FDB$. Values are in billion Euros. These results should be compared to the reported value of $FDB = 48.6$ from Table~\ref{table:allianz}.}
\label{table:lb}
\end{table}

\begin{table}[H]
\centering
\begin{tabular}{cclll}
\cline{1-4}
\multicolumn{1}{|c|}{$M=10$} 
& \multicolumn{1}{c|}{$gph = 75\%$} & \multicolumn{1}{c|}{$gph = 80\%$} 
& \multicolumn{1}{l|}{$gph = 85\%$} &  \\ 
\cline{1-4}
\multicolumn{1}{|c|}{$C_0 = 1\%$} 
& \multicolumn{1}{c|}{47.7} & \multicolumn{1}{c|}{52.5} & \multicolumn{1}{c|}{56.5} &  \\ 
\cline{1-4}
\multicolumn{1}{|c|}{$C_0 = 3\%$} 
& \multicolumn{1}{c|}{45.1} 
& \multicolumn{1}{c|}{49.8} 
& \multicolumn{1}{c|}{53.6} &  \\ 
\cline{1-4}
\multicolumn{1}{|c|}{$C_0 = 5\%$} 
& \multicolumn{1}{c|}{42.6} & \multicolumn{1}{c|}{47.0} & 
\multicolumn{1}{c|}{50.6} &  \\
\cline{1-4}
\end{tabular}%
\vspace{3pt}\caption{Lower bound for the $FDB$. Values are in billion Euros. These results should be compared to the reported value of $FDB = 48.6$ from Table~\ref{table:allianz}.}
\label{table:lb10}
\end{table}

\begin{table}[H]
\centering
\begin{tabular}{cclll}
\cline{1-4}
\multicolumn{1}{|c|}{$M=20$} 
& \multicolumn{1}{c|}{$gph = 75\%$} & \multicolumn{1}{c|}{$gph = 80\%$} 
& \multicolumn{1}{l|}{$gph = 85\%$} &  \\ 
\cline{1-4}
\multicolumn{1}{|c|}{$C_0 = 1\%$} 
& \multicolumn{1}{c|}{44.5} & \multicolumn{1}{c|}{49.3} & \multicolumn{1}{c|}{54.0} &  \\ 
\cline{1-4}
\multicolumn{1}{|c|}{$C_0 = 3\%$} 
& \multicolumn{1}{c|}{41.9} 
& \multicolumn{1}{c|}{46.6} 
& \multicolumn{1}{c|}{51.1} &  \\ 
\cline{1-4}
\multicolumn{1}{|c|}{$C_0 = 5\%$} 
& \multicolumn{1}{c|}{39.4} & \multicolumn{1}{c|}{43.8} & 
\multicolumn{1}{c|}{48.1} &  \\
\cline{1-4}
\end{tabular}%
\vspace{3pt}\caption{Lower bound for the $FDB$. Values are in billion Euros. These results should be compared to the reported value of $FDB = 48.6$ from Table~\ref{table:allianz}.}
\label{table:lb20}
\end{table}

\begin{table}[H]
\centering
\begin{tabular}{cclll}
\cline{1-4}
\multicolumn{1}{|c|}{F} 
& \multicolumn{1}{c|}{$C_0 = 1\%$} 
& \multicolumn{1}{c|}{$C_0 = 3\%$} 
& \multicolumn{1}{l|}{$C_0 = 5\%$} &  \\ 
\cline{1-4}
\multicolumn{1}{|c|}{$gph = 75\%$} 
& \multicolumn{1}{c|}{$1.3$} 
& \multicolumn{1}{c|}{$3.9$} 
& \multicolumn{1}{c|}{$6.4$} &  \\ 
\cline{1-4}
\multicolumn{1}{|c|}{$gph = 80\%$} 
& \multicolumn{1}{c|}{$1.4$} 
& \multicolumn{1}{c|}{$4.1$} 
& \multicolumn{1}{c|}{$6.9$} &  \\ 
\cline{1-4}
\multicolumn{1}{|c|}{$gph = 85\%$} 
& \multicolumn{1}{c|}{$1.5$} 
& \multicolumn{1}{c|}{$4.4$} 
& \multicolumn{1}{c|}{$7.4$} &  \\ 
\cline{1-4}
\end{tabular}%
\vspace{3pt}\caption{Values for the cross-financing term $F$. Values are in billion Euros.}
\label{table:F}
\end{table}

\begin{table}[H]
\centering
\resizebox{\textwidth}{!}{%
\begin{tabular}{cccccccccc}
\hline
\multicolumn{1}{|c|}{\textbf{12/2014}} & \multicolumn{1}{c|}{\textbf{01/2015}} & \multicolumn{1}{c|}{\textbf{02/2015}} & \multicolumn{1}{c|}{\textbf{03/2015}} & \multicolumn{1}{c|}{\textbf{04/2015}} & \multicolumn{1}{c|}{\textbf{05/2015}} & \multicolumn{1}{c|}{\textbf{06/2015}} & \multicolumn{1}{c|}{\textbf{07/2015}} & \multicolumn{1}{c|}{\textbf{08/2015}} & \multicolumn{1}{c|}{\textbf{09/2015}} \\ \hline
\multicolumn{1}{|c|}{1.08\%} & \multicolumn{1}{c|}{0.85\%} & \multicolumn{1}{c|}{0.88\%} & \multicolumn{1}{c|}{0.63\%} & \multicolumn{1}{c|}{0.78\%} & \multicolumn{1}{c|}{1.06\%} & \multicolumn{1}{c|}{1.43\%} & \multicolumn{1}{c|}{1.26\%} & \multicolumn{1}{c|}{1.40\%} & \multicolumn{1}{c|}{1.27\%} \\ \hline
\multicolumn{1}{l}{} & \multicolumn{1}{l}{} &  &  &  &  &  &  &  &  \\ \hline
\multicolumn{1}{|l|}{\textbf{10/2015}} & \multicolumn{1}{l|}{\textbf{11/2015}} & \multicolumn{1}{c|}{\textbf{12/2015}} & \multicolumn{1}{c|}{\textbf{01/2016}} & \multicolumn{1}{c|}{\textbf{02/2016}} & \multicolumn{1}{c|}{\textbf{03/2016}} & \multicolumn{1}{c|}{\textbf{04/2016}} & \multicolumn{1}{c|}{\textbf{05/2016}} & \multicolumn{1}{c|}{\textbf{06/2016}} & \multicolumn{1}{c|}{\textbf{07/2016}} \\ \hline
\multicolumn{1}{|c|}{1.22\%} & \multicolumn{1}{c|}{1.21\%} & \multicolumn{1}{c|}{1.34\%} & \multicolumn{1}{c|}{0.98\%} & \multicolumn{1}{c|}{0.75\%} & \multicolumn{1}{c|}{0.82\%} & \multicolumn{1}{c|}{0.98\%} & \multicolumn{1}{c|}{0.84\%} & \multicolumn{1}{c|}{0.66\%} & \multicolumn{1}{c|}{0.50\%} \\ \hline
\multicolumn{1}{l}{} & \multicolumn{1}{l}{} &  &  &  &  &  &  &  &  \\ \hline
\multicolumn{1}{|c|}{\textbf{08/2016}} & \multicolumn{1}{c|}{\textbf{09/2016}} & \multicolumn{1}{c|}{\textbf{10/2016}} & \multicolumn{1}{c|}{\textbf{11/2016}} & \multicolumn{1}{c|}{\textbf{12/2016}} & \multicolumn{1}{c|}{\textbf{01/2017}} & \multicolumn{1}{c|}{\textbf{02/2017}} & \multicolumn{1}{c|}{\textbf{03/2017}} & \multicolumn{1}{c|}{\textbf{04/2017}} & \multicolumn{1}{c|}{\textbf{05/2017}} \\ \hline
\multicolumn{1}{|c|}{0.52\%} & \multicolumn{1}{c|}{0.54\%} & \multicolumn{1}{c|}{0.78\%} & \multicolumn{1}{c|}{0.96\%} & \multicolumn{1}{c|}{0.96\%} & \multicolumn{1}{c|}{1.15\%} & \multicolumn{1}{c|}{0.99\%} & \multicolumn{1}{c|}{1.08\%} & \multicolumn{1}{c|}{1.10\%} & \multicolumn{1}{c|}{1.12\%} \\ \hline
\multicolumn{1}{l}{} & \multicolumn{1}{l}{} &  &  &  &  &  &  &  &  \\ \cline{1-7}
\multicolumn{1}{|c|}{\textbf{06/2017}} & \multicolumn{1}{c|}{\textbf{07/2017}} & \multicolumn{1}{c|}{\textbf{08/2017}} & \multicolumn{1}{c|}{\textbf{09/2017}} & \multicolumn{1}{c|}{\textbf{10/2017}} & \multicolumn{1}{c|}{\textbf{11/2017}} & \multicolumn{1}{c|}{\textbf{12/2017}} & \textbf{} & \textbf{} & \textbf{} \\ \cline{1-7}
\multicolumn{1}{|c|}{1.22\%} & \multicolumn{1}{c|}{1.29\%} & \multicolumn{1}{c|}{1.13\%} & \multicolumn{1}{c|}{1.26\%} & \multicolumn{1}{c|}{1.18\%} & \multicolumn{1}{c|}{1.17\%} & \multicolumn{1}{c|}{1.18\%} &  &  &  \\ \cline{1-7}
\end{tabular}%
}
\vspace{3pt}\caption{Euro base curve, 15-year spot rates.
\retwoS{Source: \cite{RFR2}}.}
\label{table:spot rates}
\end{table}

\begin{table}[H]
\centering
\begin{tabular}{|c|c|c|c|c|c|c|c|c|c|c|c|c|c|c|c|c|}
\cline{1-2} \cline{4-5} \cline{7-8} \cline{10-11} \cline{13-14} \cline{16-17}
\textbf{t} & \textbf{$P_{0,t}$} & \textbf{} & \textbf{t} & \textbf{$P_{0,t}$} & \textbf{} & \textbf{t} & \textbf{$P_{0,t}$} & \textbf{} & \textbf{t} & \textbf{$P_{0,t}$} & \textbf{} & \textbf{t} & \textbf{$P_{0,t}$} & \textbf{} & \textbf{t} & \textbf{$P_{0,t}$} \\ \cline{1-2} \cline{4-5} \cline{7-8} \cline{10-11} \cline{13-14} \cline{16-17} 
1 & 1.004 &  & 11 & 0.906 &  & 21 & 0.746 &  & 31 & 0.539 &  & 41 & 0.366 &  & 51 & 0.244 \\ \cline{1-2} \cline{4-5} \cline{7-8} \cline{10-11} \cline{13-14} \cline{16-17} 
2 & 1.005 &  & 12 & 0.889 &  & 22 & 0.726 &  & 32 & 0.519 &  & 42 & 0.351 &  & 52 & 0.234 \\ \cline{1-2} \cline{4-5} \cline{7-8} \cline{10-11} \cline{13-14} \cline{16-17} 
3 & 1.003 &  & 13 & 0.872 &  & 23 & 0.706 &  & 33 & 0.500 &  & 43 & 0.337 &  & 53 & 0.225 \\ \cline{1-2} \cline{4-5} \cline{7-8} \cline{10-11} \cline{13-14} \cline{16-17} 
4 & 0.997 &  & 14 & 0.855 &  & 24 & 0.685 &  & 34 & 0.481 &  & 44 & 0.324 &  & 54 & 0.216 \\ \cline{1-2} \cline{4-5} \cline{7-8} \cline{10-11} \cline{13-14} \cline{16-17} 
5 & 0.990 &  & 15 & 0.839 &  & 25 & 0.664 &  & 35 & 0.463 &  & 45 & 0.311 &  & 55 & 0.207 \\ \cline{1-2} \cline{4-5} \cline{7-8} \cline{10-11} \cline{13-14} \cline{16-17} 
6 & 0.979 &  & 16 & 0.824 &  & 26 & 0.643 &  & 36 & 0.446 &  & 46 & 0.299 &  & 56 & 0.199 \\ \cline{1-2} \cline{4-5} \cline{7-8} \cline{10-11} \cline{13-14} \cline{16-17} 
7 & 0.968 &  & 17 & 0.810 &  & 27 & 0.622 &  & 37 & 0.428 &  & 47 & 0.287 &  & 57 & 0.191 \\ \cline{1-2} \cline{4-5} \cline{7-8} \cline{10-11} \cline{13-14} \cline{16-17} 
8 & 0.954 &  & 18 & 0.795 &  & 28 & 0.601 &  & 38 & 0.412 &  & 48 & 0.276 &  & 58 & 0.183 \\ \cline{1-2} \cline{4-5} \cline{7-8} \cline{10-11} \cline{13-14} \cline{16-17} 
9 & 0.940 &  & 19 & 0.780 &  & 29 & 0.580 &  & 39 & 0.396 &  & 49 & 0.265 &  & 59 & 0.176 \\ \cline{1-2} \cline{4-5} \cline{7-8} \cline{10-11} \cline{13-14} \cline{16-17} 
10 & 0.923 &  & 20 & 0.764 &  & 30 & 0.559 &  & 40 & 0.381 &  & 50 & 0.254 &  & 60 & 0.169 \\ \cline{1-2} \cline{4-5} \cline{7-8} \cline{10-11} \cline{13-14} \cline{16-17} 
\end{tabular}%
\vspace{3pt}
\caption{Euro discount rates at 31.12.2017.  \retwoS{Source: \cite{RFR2}}.
}
\label{table:discount rates}
\end{table}
}

\end{document}